\documentclass[
a4paper,
aps,
pra,
showpacs,
twocolumn,
superscriptaddress]
{revtex4-2} 

\usepackage[utf8]{inputenc}  %Input what you want e.g., é, ł, a, ü
\usepackage[T1]{fontenc}     %Output what you want e.g., é, ł, a, ü
\usepackage[american]{babel}  %Do hyphenation according to american english
\usepackage[svgnames,dvipsnames]{xcolor}
\usepackage[colorlinks,
citecolor=DarkBlue,
linkcolor=DarkBlue,
linktocpage,
unicode]{hyperref}
\usepackage{amsmath,amssymb,amsthm,bm,mathtools,amsfonts,mathrsfs,bbm,dsfont} %Useful math packages
\newtheorem{theorem}{Theorem}
\newtheorem{observation}[theorem]{Observation}
\newtheorem*{remark}{No-go theorem}

\newcommand{\tr}{{\mathrm{tr}}}
\newcommand{\va}[1]{\ensuremath{(\Delta#1)^2}}

\newcommand{\ex}[1]{\ensuremath{\langle{#1}\rangle}}

\newcommand{\eins}{\mathbbm{1}}
\newcommand{\swap}{\mathbb{S}}
\renewcommand{\vr}{\ensuremath{\varrho}}
\renewcommand{\vec}[1]{\ensuremath{\boldsymbol{#1}}}
\usepackage{braket}
\usepackage[normalem]{ulem}

\begin{document}
\title{Reference-frame-independent quantum metrology}

\author{Satoya Imai}
\affiliation{Institute of Systems and Information Engineering, University of Tsukuba, Tsukuba, Ibaraki 305-8573, Japan}
\affiliation{Center for Artificial Intelligence Research (C-AIR), University of Tsukuba, Tsukuba, Ibaraki 305-8577, Japan}
\affiliation{QSTAR, INO-CNR, and LENS, Largo Enrico Fermi, 2, 50125 Firenze, Italy}
\affiliation{Naturwissenschaftlich-Technische Fakult\"at, Universit\"at Siegen, Walter-Flex-Stra{\ss}e~3, 57068 Siegen, Germany}

\author{Otfried G\"uhne}
\affiliation{Naturwissenschaftlich-Technische Fakult\"at, Universit\"at Siegen, Walter-Flex-Stra{\ss}e~3, 57068 Siegen, Germany}

\author{G\'eza T\'oth}
\affiliation{Department of Theoretical Physics, University of the Basque Country UPV/EHU, P.O. Box 644, E-48080 Bilbao, Spain}
\affiliation{EHU Quantum Center, University of the Basque Country UPV/EHU, Barrio Sarriena s/n, E-48940 Leioa, Biscay, Spain}
\affiliation{Donostia International Physics Center DIPC, Paseo Manuel de Lardizabal 4, San Sebasti\'an, E-20018, Spain}
\affiliation{IKERBASQUE, Basque Foundation for Science, E-48009 Bilbao, Spain}
\affiliation{HUN-REN Wigner Research Centre for Physics, P.O. Box 49, H-1525 Budapest, Hungary}

\date{\today}
\begin{abstract}
How can we perform a metrological task if only limited control over a quantum system is given? Here, we present systematic methods for conducting nonlinear quantum metrology in scenarios lacking a common reference frame. Our approach involves preparing multiple copies of quantum systems and then performing local measurements with randomized observables. First, we derive the metrological precision using an error propagation formula based solely on local unitary invariants, which are independent of the chosen basis. Next, we provide analytical expressions for the precision scaling in various examples of nonlinear metrology involving two-body interactions, like the one-axis twisting Hamiltonian. Finally, we analyze our results in the context of local decoherence and discuss its influences on the observed scaling.
\end{abstract}
\maketitle

%%%%%%%%%%%%%%%%%%%%%%%%%%%%%%%%%%%%%%%%%%%%%%%%%%%%%%%%%%%%%%%
\section{Introduction}
%{\it Introduction.---}
Quantum metrology involves three stages: (i) preparing an initial 
probe state, (ii) encoding a parameter $\theta$ through a state transformation, and  
(iii) measuring the transformed state to extract information about $\theta$. Each stage 
can be quantum or classical, including choices such as entangled or separable initial 
states, entangling or non-entangling state transformations, and measurements with 
multipartite operators or single-particle operators~\cite{giovannetti2004quantum, 
giovannetti2006quantum, paris2009quantum, giovannetti2011advances, toth2014quantum, 
pezze2018quantum, braun2018quantum}. The precision of parameter estimation, denoted 
as $\va{\theta}$, naturally depends on the chosen approach.

A central goal in {\it quantum} metrology is to overcome the precision reachable 
in the classical regime. In a fully classical setup, the best achievable precision 
is known as the shot-noise limit~\cite{giovannetti2006quantum, giovannetti2011advances, toth2014quantum, pezze2018quantum} $\va{\theta} \propto N^{-1}$, where $N$ is the 
number of particles in a probe system. On the other hand, the presence of 
initial entanglement in the preparation stage enables scaling beyond the 
shot-noise limit, ultimately reaching the quadratic limit: 
$\va{\theta} \propto N^{-2}$~\cite{giovannetti2006quantum, giovannetti2011advances, toth2014quantum, pezze2018quantum}. 
Furthermore, even without initial entanglement, entangling transformations can yield a scaling 
significantly better than the quadratic scaling, expressed as $\va{\theta} \propto N^{-2k+1}$ for 
integer values of $k$~\cite{luis2004nonlinear, beltran2005breaking, boixo2007generalized, 
boixo2008quantum, boixo2008quantumA, choi2008bose, boixo2009quantum, napolitano2011interaction, 
sewell2014ultrasensitive, beau2017nonlinear}, and an exponential scaling 
proposed by Roy and Braunstein~\cite{roy2008exponentially}. Note that the notion of the Heisenberg limit was originally introduced in Ref.~\cite{holland1993interferometric}, but later it was reformulated based on
a detailed accounting of relevant resources~\cite{zwierz2010general, 
zwierz2012ultimate}.

%%%%%%%%%%%%%%%%%%%%%%%%%%%%%%%%%%%%%%%%%%%%%%%%%%%%%%%%%%%%%%%
\begin{figure}[t!]
    \centering
    \includegraphics[width=0.99\columnwidth]{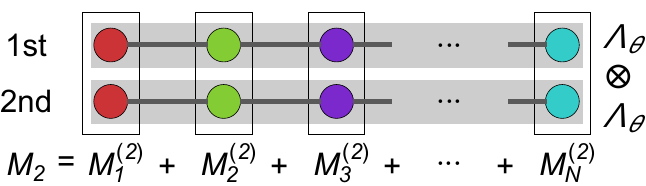}
    \caption{
    Sketch of the quantum metrology scheme presented in this paper for two copies 
    of an $N$-particle state with randomized measurements. The parameter $\theta$ is encoded
    in the 1st and 2nd copy via $\Lambda_\theta \otimes \Lambda_\theta$ 
    highlighted by gray color. 
    Subsequently, a randomized measurement $M_2 = \sum_{i=1}^N M_i^{(2)}$ is performed on the two-copy state, where 
    each local observable $M_i^{(2)}$ acts on the 1st and 2nd copy (highlighted 
    by a vertical box). This paper introduces systematic methods for evaluating 
    the precision, denoted as $\va{\theta}$, in a reference-frame-independent manner.
    }
    \label{sketch}
\end{figure}
%%%%%%%%%%%%%%%%%%%%%%%%%%%%%%%%%%%%%%%%%%%%%%%%%%%%%%%%%%%%%%%

For better accuracy, it is essential to have precise control of state preparation and parameter encoding, along with favorable measurements. In practice, however, uncontrollable experimental noise, such as magnetic field fluctuations for trapped ions or polarization rotations in optical fibers, can cause information loss in the encoding directions and disturb the calibration and alignment of the measurement settings. Such noise effects may result in the absence of a common Cartesian reference frame (associated with the $\text{SO}(3)$ or $\text{SU}(2)$ group). Also, one might consider the possibility of adversarial attacks by malicious parties that could covertly disrupt the target system, perhaps nullifying quantum advantages and preventing higher precision.

What options are available in such a non-ideal and black-box scenario where quantum control for the metrological tasks is limited? One approach is multiparameter metrology \cite{szczykulska2016multi, gessner2018sensitivity, proctor2018multiparameter, albarelli2019evaluating, carollo2019quantumness, albarelli2020perspective, liu2020quantum, demkowicz2020multi, liu2021distributed, lu2021incorporating,gorecki2022multiple} to estimate the desired
parameter as well as the error parameters in one go, but this may not work effectively if numerous unexpected variables change the system. Another approach is to establish a common reference frame, but this is known to be a resource-intensive process \cite{rudolph2003quantum, bartlett2007reference}.

%%%%%%%%%%%%%%%%%%%%%%%%%%%%%%%%%%%%%%%%%%%%%%%%%%%%%%%%%%%%%%%
\section{A no-go theorem}
%{\it A no-go theorem.---}

More specifically, let us regard a reference frame as an abstract coordinate system that allows for transforming unspeakable information (such as spin rotation) into speakable information (such as classical bits represented as $0$ or $1$), for more details see Review~\cite{bartlett2007reference}. This is a general concept of how gyroscopes and clocks give a consistent way to measure orientation and time. Here we consider reference frames in quantum systems that can be converted to another system with a different reference frame through a transformation involving an element denoted as $U \in \text{SU}(d)$. For $d=2$, the Cartesian frame in $\text{SO}(3)$ is represented by the group $\text{SU}(2)$, corresponding to spin rotations.

In fact, given an $N$-particle state, the absence of a shared reference frame between the particles excludes a quantum advantage in standard schemes of metrology. More precisely, we have the following statement.

\begin{remark}
In the absence of a common reference frame (i.e., $N$ parties do not share a reference frame), quantum metrology using 
an encoding of a parameter with the local unitary $V_L = e^{-i\theta H_L}$ for $H_L = \sum_{i=1}^N H_i$ is impossible. This holds even 
if multiple copies of the $N$-particle state are used.
\end{remark}

This can be seen as follows: If each of the $N$ parties has no 
prior knowledge about the local reference frame (related to local unitary $U_{L} = U_1 \otimes \cdots \otimes U_N$) then $k$ copies of 
the state can be represented by
\begin{equation}
    \mathcal{G}_k (\vr) = \int dU_{L} \, (U_{L} \vr U_{L}^{\dagger})^{\otimes k},
\end{equation}
where the integral is taken over the Haar measure. Hence, $\mathcal{G}_k (\vr)$ is invariant under any local unitaries, that is, $[\mathcal{G}_k (\vr), V_{L}^{\otimes k}] = 0$ for any $V_L = e^{-i\theta H_L}$. Then one cannot encode any information about the parameter $\theta$ from $V_L = e^{-i\theta H_L}$, making linear metrology impossible.

Consequently, it is natural to ask whether \textit{nonlinear} 
metrology with non-local (entangling) unitary transformations 
is still possible even without a common reference frame. Here, we introduce metrological strategies in a reference-frame-independent manner. Our main idea is to prepare multiple copies of a quantum system and employ locally randomized measurement 
observables after a parameter encoding transformation, illustrated in Fig.~\ref{sketch}.
This strategy is inspired by previous works on randomized measurements used to characterize quantum 
correlations~\cite{elben2023randomized,cieslinski2024analysing}. Although related scenarios were discussed in 
\cite{banaszek2004experimental, vsafranek2015quantum, ahmadi2015communication, xie2017quantum, fanizza2021squeezing, gorecki2022quantum}, randomized measurements were, to our knowledge, 
not employed in those studies.

In this paper, we present systematic methods to analyze the precision in quantum metrology without a common reference frame. In particular, we provide analytical expressions for the precision with several copies of the state using the error-propagation formula. Our formulation relies on local unitary invariants, which can be used for nonlinear quantum metrology rather than linear one. Then we calculate the scaling behavior of the precision for various examples and discuss how decoherence effects can worsen metrological performance.

%%%%%%%%%%%%%%%%%%%%%%%%%%%%%%%%%%%%%%%%%%%%%%%%%%%%%%%%%%%%%%%
\section{Parameter estimation and error propagation}
%{\it Parameter estimation and error propagation.---}
Consider a quantum state of $N$ $d$-dimensional particles (qudits) denoted as $\vr$ defined in the Hilbert space $\mathcal{H}_d^{\otimes N}$. Suppose that the initial state $\vr$ can be transformed by a quantum transformation $\Lambda_\theta$ to encode a parameter $\theta$: $\vr \to \vr_\theta = \Lambda_\theta (\vr)$. In the following, we focus on a unitary parameter encoding $V_\theta = e^{-i\theta H}$ with a (usually nonlocal) Hamiltonian 
$H$ as the generator of the dynamics.

The parameter $\theta$ can be estimated from a measurement observable $M$. The estimation precision can be characterized by the error-propagation formula~\cite{hotta2004quantum,escher2012quantum,toth2014quantum,pezze2018quantum,braun2018quantum}
\begin{equation} \label{eq:errorpropagationformula}
    \va{\theta}
    =\frac{\va{M}}{|\partial_\theta \ex{M}|^2},
\end{equation}
where $\va{M}\equiv \ex{M^2}-\ex{M}^2$ and $\ex{M}$ is the expectation of $M$.
Note that the precision $\va{\theta}$ depends on $\theta$ and that the above formula allows to study the situation where $\theta$ is small.

%%%%%%%%%%%%%%%%%%%%%%%%%%%%%%%%%%%%%%%%%%%%%%%%%%%%%%%%%%%%%%%
\section{Main method}
%{\it Main method.---}
Let us consider $k$ copies of the quantum state and perform a measurement on the system
\begin{equation}
    \ex{M_k} = \tr(\vr_\theta^{\otimes k} M_k),
\end{equation}
where $M_k$ acts on the $k$ copies of the state defined in $(\mathcal{H}_d^{\otimes N})^{\otimes k}$. Here we assume that the measurement observable $M_k$ is \textit{local} with respect to the parties,
\begin{equation}\label{eq:observable}
    M_k = \sum_{i=1}^N M_i^{(k)},
\end{equation}
where each $M_i^{(k)}$ acts (potentially nonlocally) on the 
$k$ copies of the $i$-th system, see also Fig.~\ref{sketch}.
As a choice of the observable $M_i^{(k)}$ 
we introduce a \textit{randomized} (twirled) measurement 
observable 
\begin{equation}
\Phi_k(\mathcal{O}) = 
    \int dU \, (U^\dagger \mathcal{O} U)^{\otimes k},
    \label{eq:ranmeasob}
\end{equation}
where $\mathcal{O}$ is a Hermitian operator defined in $\mathcal{H}_d$ and the integral is taken as the Haar random unitaries. By definition, this does not change under any unitary: $\Phi_k(V^\dagger \mathcal{O} V) = \Phi_k(\mathcal{O})$ for any unitary $V$.

For the sake of simplicity, we hereafter assume that 
$\mathcal{O}$ is traceless, $\tr(\mathcal{O}) = 0$, and normalized, $\tr(\mathcal{O}^2) = d$. We can then
formulate:

\begin{observation}\label{ob:twobody}
Consider the two copies of an $N$-qudit system with $k=2$ and 
$M_i^{(2)} = \Phi_2(\mathcal{O}_i)$. Then, the error-propagation formula leads to
\begin{equation} \label{eq:twobodyvar}
    \va{\theta}_2 
    =\frac{(d^2-1)N    -2S_1(\theta)    +2S_2(\theta)  -S_1^2(\theta) }{\left|\partial_\theta S_1(\theta)\right|^2},
\end{equation}
where $S_1(\theta) = \sum_{i=1}^N [d\tr(\vr_i^2)-1]$ for the single-particle reduced state $\vr_i = \tr_{\Bar{i}} (\vr_\theta)$ and $S_2(\theta) = \sum_{i<j} \left\{d^2\tr(\vr_{ij}^2)-1-[d\tr(\vr_i^2)-1]-[d\tr(\vr_j^2)-1] \right\}$ for the two-particle reduced state $\vr_{ij} = \tr_{\overline{ij}} (\vr_\theta)$, where all particles up to
$i,j$ are traced out.
\end{observation}

\begin{proof}
To derive Eq.~(\ref{eq:twobodyvar}), we substitute the observable in Eq.~(\ref{eq:observable}) with $M_i^{(2)} = \Phi_2(\mathcal{O}_i)$ into Eq.~(\ref{eq:errorpropagationformula}). For that, we first need to evaluate the Haar integral. In fact, assuming that
$\mathcal{O}_i$ is traceless and normalized, one has
\begin{equation} \label{eq:evaluation_kistwo}
    \Phi_2(\mathcal{O}_i)
    = \frac{1}{d^2-1}\left(d\swap_i-\eins_d^{\otimes 2}\right),
\end{equation}
where $\swap_i$ is the SWAP operator acting on both the first and second copies of the $i$-th system: $\swap \ket{x}\ket{y}=\ket{y}\ket{x}$ for details, see Refs.~\cite{werner1989quantum, emerson2005scalable, zhang2014matrix, roberts2017chaos}.  Using the 
property $\swap_i^2 = \eins_d$ and the SWAP trick $\tr\left[({X}\otimes {Y}) \swap \right] = \tr[{X} {Y}]$ for operators ${X}$ and ${Y}$, a straightforward calculation yields $\ex{M_2} = S_1(\theta)/(d^2-1)$. Similarly, we can arrive at Eq.~(\ref{eq:twobodyvar}).
\end{proof}

We have five remarks on Observation~\ref{ob:twobody}. First, the sketch of this metrological scheme is illustrated in Fig.~\ref{sketch}. Second, in general, the quantities $S_l$ for integer $l \in [1,N]$ 
are known as the $l$-body sector length~\cite{aschauer2004local, klockl2015characterizing, wyderka2020characterizing, eltschka2020maximum}, which can be associated with the purity of the $l$-particle reduced states, leading to the identity $\sum_{l=1}^N S_l(\vr) = d^N \tr(\vr^2) -1$. An important property is their invariance under local unitary transformations: $S_l(V_1 \otimes \cdots \otimes V_N \vr V_1^\dagger \otimes \cdots \otimes V_N^\dagger) = S_l(\vr)$ for any unitary $V_i$ for $i \in [1, N]$.

Third, the precision $\va{\theta}_2$ in Eq.~(\ref{eq:twobodyvar}) may remind us of the conventional spin-squeezing parameters~\cite{kitagawa1993squeezed, wineland1994squeezed, sorensen2001many, toth2007optimal, esteve2008squeezing, ma2011quantum} in the sense that the denominator relies on the reduced single-particle states and the numerator depends on both the reduced single-particle and two-particle states. In contrast to spin-squeezing parameters, the precision $\va{\theta}_2$ does not change under any parameter encoding by local unitary $V_L = e^{-i\theta H_L}$ for a local Hamiltonian~$H_L = \sum_{i=1}^N H_i$, which agrees with the no-go theorem discussed above. On the other hand, it can be changed under some unitary $V_G = e^{-i\theta H_G}$ for a non-local interaction Hamiltonian~$H_G$.

Finally, one might generalize our approach to further multicopy scenarios or other observables, perhaps resulting in additional local unitary invariants with higher degrees. However, it would be demanding to find the simple expression by the analytical evaluation of Haar integrals, and it may not necessarily lead to higher precision.

In the following, we consider four copies to derive the precision 
achievable with higher-order quantities. For the sake of simplicity, 
we focus on qubits ($d=2$) where we can, without loss of generality, take $\mathcal{O}=\sigma_z$. Then we have:

\begin{observation}\label{ob:fourcopy}
Consider the four copies of an $N$-qubit system, that is, $k=4$, $d=2$, and $M_i^{(4)} = \Phi_4(\sigma_z^{(i)})$. Then, the error-propagation formula leads to 
\begin{equation} 
    \va{\theta}_4
    \!=\! \frac{15N \!-\!20 S_1(\theta)\!+\!8 F_1(\theta)\!+\!2F_2(\theta)\!-\! 3F_1^2(\theta)}
    {3|\partial_\theta F_1(\theta)|^2},
    \label{eq:variancefour}
\end{equation}
where $F_1(\theta) = \sum_{i=1}^N [2\tr(\vr_i^2)-1]^2$ for $\vr_i = \tr_{\Bar{i}} (\vr_\theta)$ and $F_2(\theta) = \sum_{i<j}
\left\{
[\tr(T_{ij} T_{ij}^\top)]^2
+2\tr(T_{ij} T_{ij}^\top T_{ij}T_{ij}^\top)
\right\}$
with the matrix $T_{ij}$ with the elements $[T_{ij}]_{\mu \nu} = \tr(\vr_{ij} \sigma_\mu \otimes \sigma_\nu)$ for $\vr_{ij} = \tr_{\overline{ij}} (\vr_\theta)$ and $\mu, \nu = x,y,z$.
\end{observation}

The proof of Observation~\ref{ob:fourcopy} is given in Appendix~A. We remark that $F_1$ and $F_2$ are also invariants under local unitaries, where $F_2$ for $N=2$ is known as one of the Makhlin invariants \cite{makhlin2002nonlocal}. Note also that the precision $\va{\theta}_4$ depends on the one- and two-body correlations of $\varrho$ only, which can be traced back to the fact that $M_4$ is a one-body observable, and $\va{\theta}_4$ contains the variance of it.

Here it should be essential to notice that $\va{\theta}_2$ 
in Eq.~(\ref{eq:twobodyvar}) results from the two copies of 
$N$ particles, while $\va{\theta}_4$ in Eq.~(\ref{eq:variancefour}) results from the four copies of $N$ particles. 
To compare both correctly, let us introduce the gain relative to the shot-noise limit
\begin{equation}\label{eq:gainkcopy}
    G_k=\frac{1}{kN\va{\theta}_k}.
\end{equation}
Note that $G_k>1$ implies higher precision beyond the shot-noise limit.

%%%%%%%%%%%%%%%%%%%%%%%%%%%%%%%%%%%%%%%%%%%%%%%%%%%%%%%%%%%%%%%
\section{Nonlinear Hamiltonian dynamics}
%{\it Nonlinear Hamiltonian dynamics.---}
Here we show several scalings in the proximity of $\theta=0$ based on our results. For that, we consider the estimation precision in the limit of $\theta \to 0$. We present the following result:
\begin{observation}\label{ob:jx2twobodyhamscl}
    Consider that $\ket{\psi_\theta} = e^{-i\theta H}\ket{1}^{\otimes N}$ and $H=J_x^2$, where $J_x = \frac{1}{2} \sum_{i=1}^N \sigma_x^{(i)}$. Then, the gain in Eq.~(\ref{eq:gainkcopy}) is obtained as
\begin{subequations} 
\begin{align}
    \label{eq:G2nonlinear}
    \lim_{\theta\to 0}
    G_2
    &=\frac{N-1}{4},\\
    \label{eq:G4nonlinear}
    \lim_{\theta\to 0}
    G_4
    &=\frac{3(N-1)}{8},
\end{align}
\end{subequations}
for $k=2$ and $k=4,$ respectively.
\end{observation}
%%%%%%%%%%%%%%%%%%%%%%%%%%%%%%%%%%%%%%%%%%%%%%%%%%%%%%%%%%%%%%%

\begin{proof}
To prove this Observation, we need to compute all the terms $S_1(\theta), S_2(\theta), F_1(\theta)$, and $F_2(\theta)$. Since the state $\ket{\psi_\theta}$ is symmetric under exchange for any two qubits, it is sufficient to focus on one of the reduced two-qubit states and multiply its results by some factors like $N$ or $N(N-1)/2$ in the end. With the help of the result in Ref.~\cite{wang2002pairwise}, we can immediately find their explicit expressions, for details see Appendix~B. 
\end{proof}

%%%%%%%%%%%%%%%%%%%%%%%%%%%%%%%%%%%%%%%%%%%%%%%%%%%%%%%%%%%%%%%
\begin{figure}[t]
    \centering
    \includegraphics[width=0.8\columnwidth]{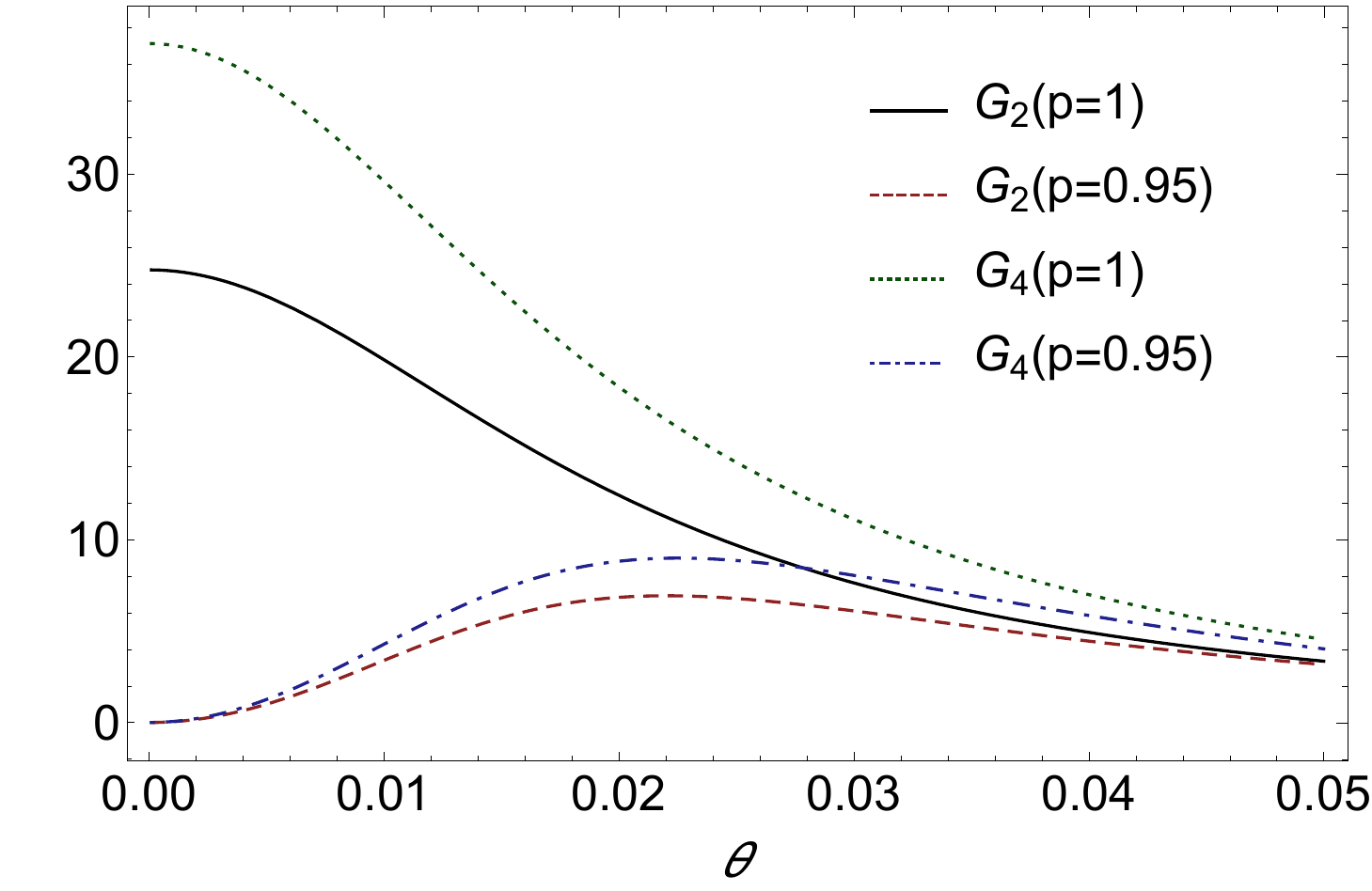}
    \caption{Sensitivity of the metrological gain defined in Eq.~(\ref{eq:gainkcopy}) to parameter shifts based on Observation~\ref{ob:jx2twobodyhamscl} in $N=100$, where $p$ denotes the noise parameter in the local depolarizing channel in Eq.~(\ref{eq:locladepola}).
    }
    \label{Fig2}
\end{figure}
%%%%%%%%%%%%%%%%%%%%%%%%%%%%%%%%%%%%%%%%%%%%%%%%%%%%%%%%%%%%
\begin{figure}[t]
    \centering
    \includegraphics[width=0.8\columnwidth]{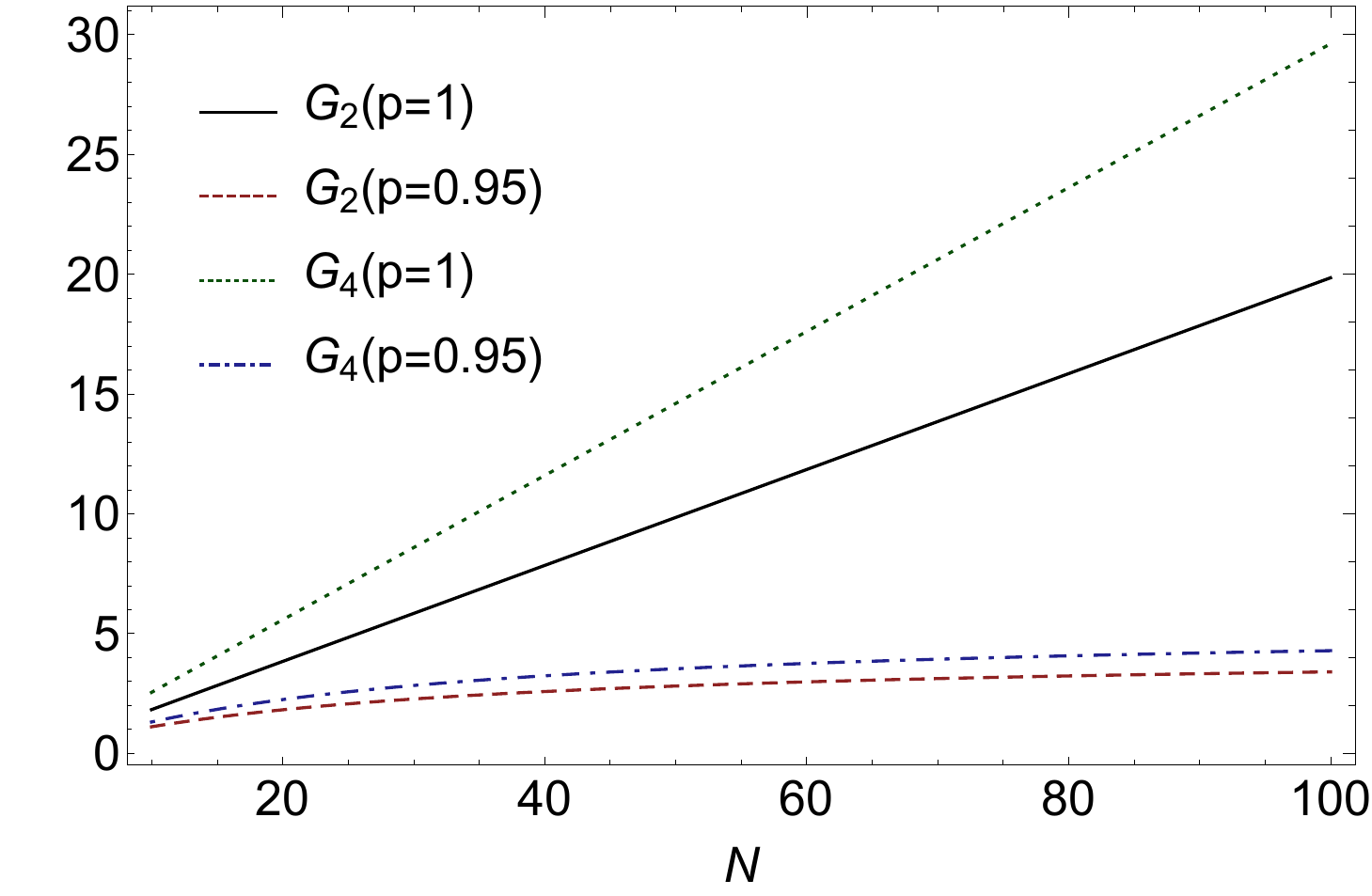}
    \caption{Growth in the metrological gain defined in Eq.~(\ref{eq:gainkcopy}) for an increasing number of particles based on Observation~\ref{ob:jx2twobodyhamscl} with a fixed $\theta=1/N$, where $p$ denotes the noise parameter in the local depolarizing channel in Eq.~(\ref{eq:locladepola}).
    }
    \label{Fig3}
\end{figure}

%%%%%%%%%%%%%%%%%%%%%%%%%%%%%%%%%%%%%%%%%%%%%%%%%%%%%%%%%%%%%%%
We remark that the nonlinear dynamics with $H=J_x^2$ is called the one-axis twisting Hamiltonian. This nonlinear dynamics is known to produce spin-squeezing entanglement \cite{kitagawa1993squeezed, ma2011quantum, pezze2018quantum} and also many-body Bell correlations \cite{plodzien2022one}. In Appendix~D, we will consider another Hamiltonian dynamics and show that the gain $G_k$ in Eq.(\ref{eq:gainkcopy}) for $k=2,4$ scales exponentially, similarly to the example of Roy and Braunstein in Ref.~\cite{roy2008exponentially}. In general, quantum entanglement generated by the unitary dynamics is necessary to achieve high values of $G_k$.

Note that our scaling $G_k \propto N$ in Eqs.~(\ref{eq:G2nonlinear}, \ref{eq:G4nonlinear}), that is, $\va{\theta} \propto 1/N^2$, was found in Ref.~\cite{rey2008many} for product states with separable measurements in the one-axis twisting Hamiltonian (moreover, the better precision $\va{\theta} \propto 1/N^3$ was discussed \cite{boixo2008quantumA, choi2008bose, boixo2009quantum}). In these
works, however, a shared reference frame between the particles was assumed, and our results show that is not needed.

%%%%%%%%%%%%%%%%%%%%%%%%%%%%%%%%%%%%%%%%%%%%%%%%%%%%%%%%%%%%%%%
\section{Decoherence}
%{\it Decoherence.---}
Here we analyze how decoherence can influence the metrological gain. As a typical decoherence model, we consider the so-called depolarizing noise channel for a quantum state $\sigma \in \mathcal{H}_d$, see Ref.~\cite{nielsen2002quantum}
\begin{equation} \label{eq:locladepola}
    \mathcal{E}_p(\sigma)
    = p \sigma + \frac{1-p}{d}\eins_d,
\end{equation}
where $0\leq p\leq 1$. In the following, let us discuss a scenario that a pure initial state is transformed by nonlinear unitary dynamics, and then each particle is \textit{locally} affected by the depolarizing channel with the same local error parameter
$\Lambda_\theta(\vr)
=\mathcal{E}_p^{\otimes N}(V_\theta \vr V_\theta^\dagger)$.
In this scenario, the precision can be obtained from the formulas given for the noiseless case in Eqs.~(\ref{eq:twobodyvar},\ref{eq:variancefour}) by replacing $S_l(\theta)$ and $F_l(\theta)$ for $l=1,2$ by
\begin{subequations}
\begin{align}
    S_l(\theta) &\to p^{2l} S_l(\theta),\\
    F_l(\theta) &\to p^{4l} F_l(\theta).
\end{align}
\end{subequations}

In Figs.~\ref{Fig2} and \ref{Fig3}, we plot the gains $G_k(p)$ for a noise parameter $p$ of Observation~\ref{ob:jx2twobodyhamscl} and compare the noiseless and a noisy case with $p=0.95$. In Fig.~\ref{Fig2}, we can find that the maximal gain in the noisy case cannot be achieved by the limit of $\theta \to 0$ unlike the noiseless case, and the optimal value of $\theta$ can be shifted depending on $p$ and $N$~\cite{lucke2011twin, urizar2013macroscopic}. In Fig.~\ref{Fig3}, we illustrate the growth in both gains for increasing particles for fixed $\theta=1/N$. We remark that the metrological gain decreases due to noise effects, where the precision in Eqs.~(\ref{eq:G2nonlinear}) and (\ref{eq:G4nonlinear}) is diminished.

%%%%%%%%%%%%%%%%%%%%%%%%%%%%%%%%%%%%%%%%%%%%%%%%%%%%%%%%%%%%%%%
\section{Collective randomization}
%{\it Collective randomization.---}
So far we have discussed the reference-frame-independent metrological scheme using locally-independent randomized observables based on Eq.~(\ref{eq:ranmeasob}). As shown in Fig.~\ref{sketch}, we have considered individual randomizations for {each local measurement on the two-copy space, as $M_2 = \sum_{i=1}^N \Phi_2(\sigma_z^{(i)})$ in Eq.~(\ref{eq:observable}).}  Here we will introduce a \textit{collective} randomization scheme that uses multilateral simultaneous rotations on all subsystems, motivated by the recent work in Ref.~\cite{imai2024collective}.

More precisely, instead of $M_2$, we consider a collective randomized observable on the two-copy system
\begin{equation}\label{eq:collectiverandomized}
    \mathcal{X}_2 = \int dU \,
    \left({U^\dagger}^{\otimes N} J_z U^{\otimes N}\right)^{\otimes 2},
\end{equation}
where $J_z = \frac{1}{2} \sum_{i=1}^N \sigma_z^{(i)}$ is known as the collective angular momentum acting on the $N$-qubit system. Here, $U$ is a local unitary on a single particle and the difference to the scheme before is that the same $U$ is applied to all parties. Then we can formulate the final result:

\begin{observation}\label{ob:collective}
Consider the two copies of an $N$-qubit system with the 
collective randomized observable $\mathcal{X}_2$. 
The error-propagation formula leads then to
\begin{equation} 
    \va{\theta}_{C_2}
    =\frac{f(N) + B(\theta)-[S_1(\theta) + K_1(\theta)]^2}
    {|\partial_\theta [S_1(\theta) + K_1(\theta)]|^2},
\end{equation}
where $f(N) = 3N(-2N+3)$ and 
\begin{align} \nonumber
    B(\theta)
    &= 2 [-S_1(\theta) -2K_1(\theta) + S_2(\theta) + K_2(\theta)]\\
    &\quad
    + K_2^\prime(\theta)
    + 16(N-1)\sum_{\mu=1}^3 \ex{J_\mu^2}.
\end{align}
Here we define that
$K_1(\theta)= \sum_{i\neq j}^N \sum_{\mu =1}^3 \tr[(\vr_{i} \otimes \vr_{j})(\sigma_\mu \otimes \sigma_\mu)]$ for $\vr_i = \tr_{\Bar{i}} (\vr_\theta)$, $K_2(\theta) =\sum_{i\neq j \neq k} \tr(T_{ij}T_{ik}^\top)$, and $K_2^\prime(\theta)=\sum_{i\neq j \neq k \neq l} \tr(T_{ik}T_{jl}^\top)$ with the matrix $T_{ij}$ with the elements $[T_{ij}]_{\mu \nu} = \tr(\vr_{ij} \sigma_\mu \otimes \sigma_\nu)$ for $\vr_{ij} = \tr_{\overline{ij}} (\vr_\theta)$.
\end{observation}

The proof of Observation~\ref{ob:collective} is given in Appendix~C. In order to distinguish with the precision $\va{\theta}_{2}$ in Observation~\ref{ob:twobody}, we add the subscript $C_2$ in the precision $\va{\theta}_{C_2}$. By definition,  $\va{\theta}_{C_2}$ does not change under any \textit{collective} local unitary $V^{\otimes N}$.

We apply this result to nonlinear metrology with the one-axis twisting Hamiltonian $H=J_x^2$. From a similar calculation with Observation~\ref{ob:jx2twobodyhamscl} (see Appendix~B), the gain $G_{C_2}=[2N \va{\theta}_{C_2}]^{-1}$ is obtained as
\begin{equation}
    \lim_{\theta \to 0} G_{C_2} = \frac{N^2(N-1)}{2 N (3 N-5)+8}.
\end{equation}
This scaling is the same as in Eq.~(\ref{eq:G2nonlinear}).

%%%%%%%%%%%%%%%%%%%%%%%%%%%%%%%%%%%%%%%%%%%%%%%%%%%%%%%%%%%%%%%
\section{Experimental implications}
%{\it Experimental implications.---}
The one-axis twisting Hamiltonian $H=J_x^2$ can be realized in several physical systems. In trapped cold atoms, even addressing the particles is possible \cite{molmer1999multiparticle, bohnet2016quantum}. In a Bose-Einstein condensate, this can be realized by collisional interactions between atoms \cite{sorensen2001many}. Also, the generation of Haar random unitaries can be implemented by stochastic quantum walks on the integrated photonic chips \cite{banchi2017driven,tang2022generating} or a polarization controller operating in scrambling mode \cite{wyderka2023complete}. It is also possible to realize such operations in parallel in an ensemble of many multi-qudit quantum systems \cite{toth2007efficient}.

We note that the Haar integral on $\Phi_k(\mathcal{O})$ as in Eq.~(\ref{eq:ranmeasob}) can be implementable by measuring permutation operators, such as Eq.~(\ref{eq:evaluation_kistwo}) for $k=2$. More generally, the Schur-Weyl duality implies that any observable $\Phi_k(\mathcal{O})$ can be expressed as linear combination of permutation operators~\cite{roberts2017chaos}, i.e., $\Phi_k(\mathcal{O}) = \sum_{\pi \in {\rm Sym}_k} c_{\pi} W_{\pi}$, where $c_{\pi}$ is coefficients that depends on $\mathcal{O}$ and $W_{\pi}$ is a permutation operator corresponding to a permutation $\pi$ in the symmetric group ${\rm Sym}_k$ of order $k$.

%%%%%%%%%%%%%%%%%%%%%%%%%%%%%%%%%%%%%%%%%%%%%%%%%%%%%%%%%%%%%%%
\section{Conclusion}
%{\it Conclusion.---}
We have proposed (collective) reference-frame-independent metrological schemes with two and four copies of a quantum state. Our formulation is invariant under (collective) local unitaries, which enables us to perform nonlinear metrology with non-local transformations. We analytically computed the precision for several relevant cases, showing that they go beyond the shot-noise limit. 

There are several research directions in which our work can be generalized. First, it would be interesting to show analytical examples of various types of scaling with the nonlinear Hamiltonian $H=J_x^k$ \cite{boixo2007generalized,  boixo2008quantum,  napolitano2011interaction}. Second, by the spirit of spin squeezing, finding metrologically meaningful uncertainty relations from random measurements may give fundamental limitations of precision. Finally, our method may encourage the further development of parameter estimation tasks in terms of multicopy metrology \cite{toth2020activating}, multiparameter scenarios \cite{albarelli2020perspective, demkowicz2020multi}, or temperature estimation in quantum thermodynamics \cite{mehboudi2019thermometry}.

%%%%%%%%%%%%%%%%%%%%%%%%%%%%%%%%%%%%%%%%%%%%%%%%%%%%
\section{Acknowledgments}
%{\it Acknowledgments.---}
We would like to thank
Francesco Albarelli,
Iagoba Apellaniz,
Jan Lennart Bönsel,
Stefan Nimmrichter,
Luca Pezz\'e,
Augusto Smerzi,
R\'obert Tr\'enyi,
Nikolai Wyderka,
and
Benjamin Yadin
for discussions.

This work was supported by
the DAAD,
the Deutsche Forschungsgemeinschaft (DFG, German Research Foundation, project numbers 447948357,
440958198, and 56343716),
the Sino-German Center for Research Promotion (Project M-0294),
the ERC (Consolidator Grant 683107/TempoQ),
the German Federal Ministry of Research, Technology and Space ((Project QuKuK,
Grant No. 16KIS1618K and Project BeRyQC, Grant
No. 13N17292),
National Research, Development and Innovation Office of Hungary (NKFIH) (Grant No. 2019-2.1.7-ERA-NET-2021-00036, Advanced Grant No. 152794),
and the National Research, Development and Innovation Office of Hungary (NKFIH) within the Quantum Information National Laboratory of Hungary.
We acknowledge the support of the EU 
%(COST Action CA15220, QuantERA CEBBEC,
(QuantERA MENTA, QuantERA QuSiED),
the Spanish MCIU 
%(Grant No. PCI2018-092896, 
(No. PCI2022-132947), and the Basque Government 
%(Grant No. IT986-16, 
(No. IT1470-22). 
We acknowledge the support of the Grant~No.~PID2021-126273NB-I00 funded by MCIN/AEI/10.13039/501100011033 and by "ERDF A way of making Europe".  
We thank the "Frontline" Research Excellence Programme of the NKFIH (Grant No. KKP133827).
G.~T. acknowledges a  Bessel Research Award of the Humboldt Foundation.
We acknowledge Project no. TKP2021-NVA-04, which has been implemented with the support provided by the Ministry of Innovation and Technology of Hungary from the National Research, Development and Innovation Office (NKFIH), financed under the TKP2021-NVA funding scheme. We thank support from Horizon Europe programme HORIZON-CL4-2022-QUANTUM-02-SGA via the project 101113690 (PASQuanS2.1) and JST ASPIRE (JPMJAP2339).

%%%%%%%%%%%%%%%%%%%%%%%%%%%%%%%%%%%%%%%%%%%%%%%%%%%%%%%%%%%%%%%%%%%%%%%%%%%%%%%%%%%%%%%%%%%%%%%%%
\appendix
\pagenumbering{arabic}
\addtocounter{theorem}{-3}

%%%%%%%%%%%%%%%%%%%%%%%%%%%%%%%%%%%%%%%%%%%%%%%%%%%%%%%%%%%%%%%%%%%%%%%%%%%%%%%%%%%%%%%%%%%%%%%%%
\section{Proof of Observation~\ref{ob:fourcopy}}\label{ap:fourcopy}
\begin{observation}
Consider the four copies of an $N$-qubit system, that is, $k=4$, $d=2$, and $M_i^{(4)} = \Phi_4(\sigma_z^{(i)})$. Then, the error-propagation formula leads to 
\begin{equation} 
    \va{\theta}_4
    = \frac{15N -20 S_1(\theta)+8 F_1(\theta)+2F_2(\theta)- 3F_1^2(\theta)}
    {3|\partial_\theta F_1(\theta)|^2},
\end{equation}
where $F_1(\theta) = \sum_{i=1}^N [2\tr(\vr_i^2)-1]^2$ for $\vr_i = \tr_{\Bar{i}} (\vr_\theta)$ and 
\begin{equation}
    F_2(\theta) = \sum_{i<j}
\left\{
[\tr(T_{ij} T_{ij}^\top)]^2
+2\tr(T_{ij} T_{ij}^\top T_{ij}T_{ij}^\top)
\right\},
\end{equation}
with the matrix $T_{ij}$ with the elements $[T_{ij}]_{\mu \nu} = \tr(\vr_{ij} \sigma_\mu \otimes \sigma_\nu)$ for $\vr_{ij} = \tr_{\overline{ij}} (\vr_\theta)$ and $\mu, \nu = x,y,z$.
\end{observation}

\begin{proof}
We begin by evaluating the form of $\ex{M_4}$ as follows
\begin{align}
    \ex{M_4}
    &= \sum_{i=1}^N \tr\left[
    \vr_i^{\otimes 4} \Phi_4(\sigma_z^{(i)})
    \right]\nonumber\\
    &= \frac{1}{2^4} \sum_{i=1}^N
    \sum_{a,b,c,d}
    r_a^{(i)}
    r_b^{(i)}
    r_c^{(i)}
    r_d^{(i)}
    \mathcal{I}(a,b,c,d)\nonumber\\
    &=\frac{1}{5} F_1(\theta).
\end{align}
Here in the first line, we use that $\Phi_4(\sigma_z^{(i)})$ only acts on the four copies of the $i$-th system. In the second line, we denote that for $r_\mu^{(i)} = \tr[\vr_i \sigma_\mu^{(i)}]$ for $\mu=a,b,c,d=x,y,z$ and
\begin{equation}
    \mathcal{I}(a,b,c,d)
    =
    \int dU\,
    Z_{U,a}^{(i)}
    Z_{U,b}^{(i)}
    Z_{U,c}^{(i)}
    Z_{U,d}^{(i)},
\end{equation}
where $Z_{U,a}^{(i)} =\tr[\sigma_a^{(i)} U^\dagger \sigma_z^{(i)} U]$. In the third line, we apply the following formula
\begin{equation}
    \mathcal{I}(a,b,c,d)
    =\frac{16}{15}
    \left(
    \delta_{a,b} \delta_{c,d}
    + \delta_{a,c} \delta_{b,d}
    + \delta_{a,d} \delta_{b,c}
    \right),
\end{equation}
given in Ref.~\cite{wyderka2023complete} and introduce the fourth-order quantity
\begin{equation}
    F_1(\theta) = \sum_{i=1}^N r_i^4,
\end{equation}
where $r_i^2 = \sum_{\mu=x,y,z} [r_\mu^{(i)}]^2= 2\tr[\vr_i^2]-1$. 

Next, we will evaluate the expression of the variance: $\va{M_4} = \ex{M_4^2} - \ex{M_4}^2$, where
\begin{equation}
    \ex{M_4^2} = \sum_{i} \ex{\Phi_4(\sigma_z^{(i)})^2} + \sum_{i \neq j } \ex{\Phi_4(\sigma_z^{(i)}) \Phi_4(\sigma_z^{(j)})}.
\end{equation}
The first term in $\ex{M_4^2}$ is given by
\begin{align}
    \nonumber
    &\sum_{i} \ex{\Phi_4(\sigma_z^{(i)})^2}
    \\
    =&\sum_{i}
    \tr\left\{
    \left[\vr_i^{\otimes 4} \Phi_4(\sigma_z^{(i)})\right]_X
    \otimes
    [\Phi_4(\sigma_z^{(i)})]_Y
    \,
    \swap_{X,Y}
    \right\}\nonumber\\
    =&\frac{1}{2^4}\sum_{i}
    \int dU_X \int dU_Y
    \left[
    \tr\left(
    \chi_x^{(i)} \otimes \upsilon_y^{(i)}
    \swap_{x, y}
    \right)
    \right]^4\nonumber\\
    =&\frac{1}{2^4 \cdot 2^4}
    \sum_{i}
    \int dU_X \int dU_Y
    \left[
    \sum_{\alpha}
    \mathcal{Z}_{U_X, \alpha}^{(i)}
    Z_{U_Y,\alpha}^{(i)}
    \right]^4
    \nonumber\\
    =&\frac{1}{2^4\cdot 2^4}
    \sum_{i}
    \sum_{\alpha,\beta,\gamma,\delta}
    \mathcal{I}(\alpha,\beta,\gamma,\delta)
    \sum_{j=0}^4 \mathcal{C}_{j}^{(i)}\nonumber\\
    =&\frac{1}{5 \times 15}
    \left[
    15N-20 S_1(\theta)+8 F_1(\theta)
    \right].
\end{align}
Here in the first equality, we divide the squared term into two different spaces $X, Y$ inversely using the SWAP trick mentioned in the proof of Observation~\ref{ob:twobody} in the main text: $\tr(XY^2) = \tr[(XY \otimes Y)\swap]$. The SWAP $\swap_{X, Y}$ acts on the eight-qubit system, where each system $X=\{x_1, x_2, x_3, x_4\}$ and $Y=\{y_1, y_2, y_3, y_4\}$ is the four-copy of a single-qubit system.

In the second equality, we use that the SWAP operator in many qubits can be realized by the SWAP operators in individual qubits~\cite{ekert2002direct}, that is,
\begin{equation}
    \swap_{X, Y} =
\swap_{x_1, y_1}
\otimes
\swap_{x_2, y_2}
\otimes
\swap_{x_3, y_3}
\otimes
\swap_{x_4, y_4}.
\end{equation}
Also, we denote that
\begin{subequations}
    \begin{align}
        \chi_x^{(i)}
        &=
        U_X^\dagger \sigma_z^{(i)} U_X + \sum_{a=x,y,z} r_a^{(i)} \sigma_a^{(i)} U_X^\dagger \sigma_z^{(i)} U_X,
        \\
        \upsilon_y^{(i)}
        &= U_Y^\dagger \sigma_z^{(i)} U_Y,
    \end{align}
\end{subequations}
where $r_a^{(i)} = \tr[\vr_{i} \sigma_a^{(i)}]$.

In the third equality, we apply the formulas
\begin{subequations}
    \begin{align}
    \swap
    &= \frac{1}{2}\left(
    \eins_{2}^{\otimes 2}
    +
    \sum_{\alpha=x,y,z}
    \sigma_\alpha
    \otimes
    \sigma_\alpha
    \right),
    \\
    \label{eq:pauliproperty}
    \sigma_p \sigma_q
    &= \delta_{p,q}\eins_2 + i\sum_{r=x,y,z}
    \varepsilon_{p,q,r} \sigma_{r},
    \end{align}
\end{subequations}
with the Kronecker-delta symbol $\delta_{p,q}$ and the Levi-Civita symbol $\varepsilon_{p,q,r}$, and denote that
\begin{equation}
    \mathcal{Z}_{U_X, \alpha}^{(i)}
    =
    Z_{U_X,\alpha}^{(i)}
    +i
    \sum_{a, k =x,y,z}
    \varepsilon_{\alpha, a, k}
    r_a^{(i)}
    Z_{U_X, k}^{(i)},
\end{equation}
where $Z_{U_X,\alpha}^{(i)} =\tr[\sigma_\alpha^{(i)} U_X^\dagger \sigma_z^{(i)} U_X]$.

In the fourth equality, we denote that
\begin{equation}
    \int dU_X\,
    \mathcal{Z}_{U_X, \alpha}^{(i)}
    \mathcal{Z}_{U_X, \beta}^{(i)}
    \mathcal{Z}_{U_X, \gamma}^{(i)}
    \mathcal{Z}_{U_X, \delta}^{(i)}
    =\sum_{j=0}^4 \mathcal{C}_{j}^{(i)},
\end{equation}
where $\alpha,\beta,\gamma,\delta=x,y,z$ and the label $j$ in $\mathcal{C}_{j}^{(i)}$ represents the number of times the imaginary unit $i$ is multiplied.

In the final equality, we indeed evaluate all the terms in $\mathcal{C}_{j}^{(i)}$ and simplify the expression. Note that $\mathcal{C}_{1}^{(i)}=\mathcal{C}_{3}^{(i)}=0$ for any $i$ due to the properties of the Kronecker delta and Levi-Civita symbol.

Let us continue the computation of the variance. The second term in $\ex{M_4^2}$ can be given by
\begin{align}
    \nonumber
    &\sum_{i \neq j}
    \ex{\Phi_4(\sigma_z^{(i)}) \Phi_4(\sigma_z^{(j)})}
    \\
    = &\sum_{i \neq j}
    \tr[\vr_{ij}^{\otimes 4} \Phi_4(\sigma_z^{(i)}) \Phi_4(\sigma_z^{(j)})]\nonumber\\
    = &\frac{1}{4^4}
    \sum_{i \neq j}
    \sum_{
    \textit{\textbf{a}}, \textit{\textbf{b}}}
    t_{a_1 b_1}^{(i j)}
    t_{a_2 b_2}^{(i j)}
    t_{a_3 b_3}^{(i j)}
    t_{a_4 b_4}^{(i j)}
    \mathcal{I}(\textit{\textbf{a}})
    \mathcal{I}(\textit{\textbf{b}})\nonumber\\
    = &\frac{2}{5 \times 15}F_2(\theta).
\end{align}
In the second equality, we denote that $\textit{\textbf{a}} = (a_1, a_2, a_3, a_4)$, $\textit{\textbf{b}} = (b_1, b_2, b_3, b_4)$, and $t_{a_p b_p}^{(i j)} = \tr(\vr_{i j} \sigma_{a_p} \otimes \sigma_{b_p})$ for $a_p, b_p =x,y,z$. In the third equality, we use that the sector length can be given by $S_2(\theta) = \sum_{i<j} \tr(T_{ij}T_{ij}^\top) = \sum_{i<j} [4\tr(\vr_{ij}^2)-1-S_1(\theta)]$ with the matrix $[T_{ij}]_{ab} = t_{a b}^{(i j)}$ and introduce the fourth-order two-body quantity
\begin{equation}
    F_2(\theta) = \sum_{i<j}
    \left\{
    [\tr(T_{ij} T_{ij}^\top)]^2
    +2\tr(T_{ij} T_{ij}^\top T_{ij}T_{ij}^\top)
    \right\}.
\end{equation}
Hence we can complete the proof.
\end{proof}
%%%%%%%%%%%%%%%%%%%%%%%%%%%%%%%%%%%%%%%%%%%%%%%%%%%%%%%%%%%%%%%
\section{Derivation of Observation~\ref{ob:jx2twobodyhamscl}}\label{ap:jx2twobodyhamscl}

\begin{observation}
    Consider that $\ket{\psi_\theta} = e^{-i\theta H}\ket{1}^{\otimes N}$ and $H=J_x^2$, where $J_x = \frac{1}{2} \sum_{i=1}^N \sigma_x^{(i)}$. Then, the gain in Eq.~(\ref{eq:gainkcopy}) in the main text is obtained as
\begin{subequations} 
\begin{align}
    \label{eq:app:g2}
    \lim_{\theta\to 0}
    G_2
    &=\frac{N-1}{4},\\
    \label{eq:app:g4}
    \lim_{\theta\to 0}
    G_4
    &=\frac{3(N-1)}{8},
\end{align}
\end{subequations}
for $k=2$ and $k=4,$ respectively.
\end{observation}

Here we will give the the explicit expressions of $S_1(\theta), S_2(\theta), F_1(\theta)$, and $F_2(\theta)$. Let us begin by recalling that all the single-qubit and two-qubit reduced states $\vr_i$ and $\vr_{ij}$ are the same for $i,j=1,\ldots,N$. Then we denote that $\vr_i = \vr_1$ and $\vr_{ij} = \vr_2$. According to the result in Ref.~\cite{wang2002pairwise}, for $r_\mu = \tr(\vr_1 \sigma_\mu)$ and $t_{\mu \nu} = \tr(\vr_2 \sigma_\mu \otimes \sigma_\nu)$, we can have
\begin{subequations} 
\begin{align}
    r_x &= r_y = 0,
    \quad
    r_z = -\cos^{N-1}(\theta),
    \\
    t_{xx} &= t_{xz} = t_{zx} = t_{yz} = t_{zy} = 0,
    \\
    t_{yy} &= \frac{1}{2} \left[1-\cos^{N-2}(2 \theta )\right],
    \\
    t_{zz} &= \frac{1}{2} \left[\cos^{N-2}(2 \theta )+1\right],
    \\
    t_{xy} &= t_{yx} = \sin (\theta ) \cos^{N-2}(\theta ).
\end{align}
\end{subequations} 
This yields
\begin{subequations} 
\begin{align}
    S_1(\theta)
    &= N \sum_{\mu = x,y,z} r_\mu^2
    = N \cos^{2 N-2}(\theta ),
    \\
    \nonumber
    S_2(\theta)
    &= \frac{N(N-1)}{2} \sum_{\mu,\nu = x,y,z} t_{\mu \nu}^2
    \\ \nonumber
    &= \frac{N(N-1)}{4} \Big[\cos^{2 (N-2)}(2 \theta )
    \\
    &\quad+4 \sin^2(\theta ) \cos^{2 N-4}(\theta )+1\Big],
    \\
    F_1(\theta)
    &= N \Big[\sum_{\mu = x,y,z} r_\mu^2\Big]^2 = N \cos^{4 N-4}(\theta ),
    \\
    F_2(\theta)
    &= \frac{N(N-1)}{2} \Bigg\{
    \Big[\sum_{\mu,\nu = x,y,z} t_{\mu \nu}^2\Big]^2
    \nonumber
    \\
    &\quad+ 2 \sum_{\mu. \nu, \xi, \kappa = x,y,z} t_{\mu \nu}t_{\mu \kappa}t_{\xi \nu}t_{\xi \kappa}
    \Bigg\}
    \nonumber
    \\
    &= \frac{N(N-1)}{4} \Bigg\{
    \cos^{4 (N-2)}(2 \theta )
    \nonumber
    \\
    &\qquad
    -8 \sin^2(\theta ) \cos^{2 N-4}(\theta ) \cos^{N-2}(2 \theta )
    \nonumber
    \\
    \nonumber
    &\qquad
    +\cos^{2 (N-2)}(2 \theta ) \left[8 \sin^2(\theta ) \cos^{2 N-4}(\theta )+4\right]
    \\
    &\qquad
    +\left(4 \sin^2(\theta ) \cos^{2 N-4}(\theta )+1\right)^2\Bigg\}.
\end{align}
\end{subequations} 
Substituting these expressions into the precisions $\va{\theta}_2$ in Eq.~(\ref{eq:twobodyvar}) and $\va{\theta}_4$ in Eq.~(\ref{eq:variancefour}) in the main text and taking the limit $\theta \to 0$, we can arrive at the results in Eqs.~(\ref{eq:app:g2}, \ref{eq:app:g4}).

%%%%%%%%%%%%%%%%%%%%%%%%%%%%%%%%%%%%%%%%%%%%%%%%%%%%%%%%%%%%%%%
\section{Proof of Observation~\ref{ob:collective}}\label{ap:collective}

\begin{observation}
Consider the two copies of an $N$-qubit system with the 
collective randomized observable $\mathcal{X}_2$. 
The error-propagation formula leads to
\begin{equation} 
    \va{\theta}_{C_2}
    =\frac{f(N) + B(\theta)-[S_1(\theta) + K_1(\theta)]^2}
    {|\partial_\theta [S_1(\theta) + K_1(\theta)]|^2},
\end{equation}
where $f(N) = 3N(-2N+3)$ and 
\begin{align} \nonumber
    B(\theta)
    &= 2 [-S_1(\theta) -2K_1(\theta) + S_2(\theta) + K_2(\theta)]
    + K_2^\prime(\theta)
    \\
    &\quad
    + 16(N-1)\sum_{\mu=1}^3 \ex{J_\mu^2}.
\end{align}
Here we define that $K_1(\theta)= \sum_{i\neq j}^N \sum_{\mu =1}^3 \tr[(\vr_{i} \otimes \vr_{j})(\sigma_\mu \otimes \sigma_\mu)]$ for $\vr_i = \tr_{\Bar{i}} (\vr_\theta)$ and
\begin{equation}
    K_2(\theta) \!=\! \! \! \! \sum_{i\neq j \neq k} \! \! \! \!
    \tr(T_{ij}T_{ik}^\top),
    \quad
    K_2^\prime(\theta) \!=\! \! \! \! \sum_{i\neq j \neq k \neq l} \! \! \! \!
    \tr(T_{ik}T_{jl}^\top),
\end{equation}
with the matrix $T_{ij}$ with the elements $[T_{ij}]_{\mu \nu} = \tr(\vr_{ij} \sigma_\mu \otimes \sigma_\nu)$ for $\vr_{ij} = \tr_{\overline{ij}} (\vr_\theta)$ and $\mu, \nu = x,y,z$.
\end{observation}

\begin{proof}
We begin by denoting
\begin{align} \nonumber
    \mathcal{X}_2
    &= \int dU \,
    \left({U^\dagger}^{\otimes N} J_z U^{\otimes N}\right)^{\otimes 2}
    \\ \nonumber
    &=\frac{1}{4} \sum_{i,j=1}^N
    \int dU \,
    \left({U^\dagger} \sigma_z^{(i)} U\right)
    \otimes
    \left({U^\dagger} \sigma_z^{(j)} U\right)
    \\
    &\equiv
    \frac{1}{4} \sum_{i,j=1}^N \Phi_{ij},
\end{align}
where
\begin{equation} \label{eq:Phiij}
    \Phi_{ij}
    =\frac{1}{3}(2\swap_{ij}-\eins)
    =\frac{1}{3}\sum_{\mu=x,y,z}
    \sigma_\mu^{(i)} \otimes \sigma_\mu^{(j)}.
\end{equation}
Here the $i$ denotes the $i$-th system in the first copy and the $j$ denotes the $j$-th system in the second copy. Then we can immediately evaluate the form of $\ex{\mathcal{X}_2}$ as follows
\begin{align}\nonumber
    \ex{\mathcal{X}_2}
    &= \tr(\vr^{\otimes 2} \mathcal{X}_2)
    \\ \nonumber
    &= \frac{1}{4} \sum_{i,j=1}^N \tr[(\vr_i \otimes \vr_j) \Phi_{ij}]
    \\
    &= \frac{1}{3\times 4}\left[S_1(\theta) + K_1(\theta) \right],
\end{align}
where $\vr_{i} = \tr_{\overline{i}}(\vr_\theta)$ is the single-particle reduced state.

Next, we will evaluate the form of the variance: $\va{\mathcal{X}_2} = \ex{\mathcal{X}_2^2} - \ex{\mathcal{X}_2}^2$. Here, a straightforward calculation yields
\begin{align} \nonumber
    \mathcal{X}_2^2
    &\!=\!
    \frac{1}{16}\Bigg\{ \!
    \sum_{i} \Phi_{ii}^2
    \!+\!
    \sum_{i\neq j}\left[
    \Phi_{ii}\Phi_{jj}
    +4\Phi_{ii}\Phi_{ij}
    +2\Phi_{ij}^2
    \right]
    \\
    &\!+\!
    \! \!
    \sum_{i\neq j\neq k}\left[
    4\Phi_{ij}\Phi_{ik}
    +2\Phi_{ii}\Phi_{jk}
    \right]
    \!+\!
    \sum_{i\neq j\neq k \neq l}
    \Phi_{ij}\Phi_{kl}
    \Bigg\},
\end{align}
where we used that $\Phi_{ij} = \Phi_{ji}$. To find the explicit form of $\va{\mathcal{X}_2}$, we must evaluate all these expectations. A long calculation leaves us with
\begin{subequations} 
\begin{align}
    &\sum_{i} \ex{\Phi_{ii}^2}
    = \frac{1}{9}[3N-2S_1(\theta)],
    \\
    &\sum_{i\neq j}\ex{\Phi_{ii}\Phi_{jj}}
    =\frac{2}{9}S_2(\theta),
    \\
    &\sum_{i\neq j}\ex{\Phi_{ii}\Phi_{ij}}
    =\frac{1}{9}\left[    4\sum_{\mu=x,y,z}\tr(\vr_\theta J_\mu^2) - 3 N    \right],
    \\
    &\sum_{i\neq j}\ex{\Phi_{ij}^2}
    =    \frac{1}{9}\left[   3N(N-1)-2K_1(\theta)    \right],
    \\
    &\sum_{i\neq j \neq k}\ex{\Phi_{ij}\Phi_{ik}}
    =\frac{(N-2)}{9}\left[ 4 \!\!
    \sum_{\mu=x,y,z} \tr(\vr_\theta J_\mu^2) - 3N    \right],
    \\
    &\sum_{i\neq j \neq k} \ex{\Phi_{ii}\Phi_{jk}}
    =\frac{1}{9}\sum_{i\neq j \neq k} K_2(\theta),
    \\
    &\sum_{i\neq j\neq k \neq l} \ex{\Phi_{ij}\Phi_{kl}}
    =\frac{1}{9}\sum_{i\neq j \neq k} K_2^\prime(\theta).
\end{align}
\end{subequations}
Summarizing these terms, we can complete the proof. Here, it might be useful for some readers to note that
\begin{equation}
    \Phi_{ij}^2
    \!=\! \frac{1}{9} \!
    \sum_{\mu, \nu=x,y,z} \!
    \sigma_\mu^{(i)}\sigma_\nu^{(i)}
    \otimes
    \sigma_\mu^{(j)}\sigma_\nu^{(j)}
    \!=\! \frac{1}{9}\left(5\eins -4\swap_{ij}\right),
\end{equation}
where we recall Eq.~(\ref{eq:Phiij}) and use the property in Eq.~(\ref{eq:pauliproperty}) and
$\sum_{\mu, \nu = x,y,z} \varepsilon_{\mu, \nu, \xi}\varepsilon_{\mu, \nu, \xi^\prime}
= 2 \delta_{\xi, \xi^\prime}$.
\end{proof}

\noindent
\textbf{Remark.}
Consider a permutationally invariant state: $\vr = \swap_{ij} \vr \swap_{ij}$ for all $i,j \in \{1,2,\ldots, N\}$ with $i\ne j$, for details see Ref.~\cite{toth2009entanglement}. Since $\swap (X \otimes Y) \swap = Y \otimes X$ for operators $X,Y$, any permutationally invariant state $\vr$ has the same Bloch vector $\vec{r} = \vec{r}_i$ and the same matrix $T = T_{ij}$ for all $i\ne j$ with $[T]_{\mu \nu} = [T]_{\nu \mu}$. Thus we can simplify the expressions
\begin{subequations}
\begin{align}
&S_1(\theta) + K_1(\theta)
= N^2 r^2(\theta),\\
&B(\theta)
= N \Big\{r^2 (2-4 N)+(N-1) \big[(3+N(N-3))t_2(\theta)
\nonumber
\\
&\quad \quad
+4 (N-1) t_1(\theta) + 12\big]\Big\},
\end{align}
\end{subequations}
where we denoted that $r^2(\theta) = |\vec{r}|^2$, $t_1(\theta) = \tr(T)$, and $t_2(\theta) = \tr(T^2)$. To find the above simplification, we used 
\begin{subequations}
\begin{align}
    S_1 &= N r^2,
    \\
    K_1 &= N(N-1) r^2,
    \\
    S_2 &= \frac{N(N-1)}{2}t_2,
    \\
    K_2 &= N(N-1)(N-2)t_2,
    \\
    K_2^\prime &= N(N-1)(N-2)(N-3)t_2,
    \\
    \sum_{\mu=x,y,z}
    \ex{J_\mu^2}
    &=\frac{1}{4}[3N+N(N-1)t_1],
\end{align}
\end{subequations}
where we abbreviated the notation of $\theta$.

%%%%%%%%%%%%%%%%%%%%%%%%%%%%%%%%%%%%%%%%%%%%%%%%%%%%%%%%%%%%%%%
\section{Exponential scaling}\label{ap:fullbodyhamscl}
Here we formulate the exponential scaling:
\begin{observation}\label{ob:fullbodyhamscl}
    Consider that $\ket{\psi_\theta^{n}} = e^{-i\theta H}\ket{0}^{\otimes N}$ and $H = Q_n$ for $n=1,2$ such that $Q_1 + iQ_2 = (\sigma_x + i \sigma_y)^{\otimes N}$. Then the gain in Eq.~(\ref{eq:gainkcopy}) in the main text is obtained as
\begin{subequations}
\begin{align}
    \lim_{\theta\to 0}
    G_2 &= \frac{4^N}{N+1},\\
    \lim_{\theta\to 0}
    G_4 &= \frac{3 \times 2^{2 N+1}}{3 N+1},
\end{align}
\end{subequations}
for $k = 2$ and $k = 4$, respectively.
\end{observation}

\noindent
\textbf{Remark.}
The Hamiltonian $Q_n$ for $n=1,2$ is the same as the Hermitian operator used in Mermin-type inequalities~\cite{mermin1990extreme,roy2005multipartiteseparability,toth2005entanglement,gachechiladze2016extreme}. This Hamiltonian model was considered by Roy and Braunstein in Ref.~\cite{roy2008exponentially}.

%%%%%%%%%%%%%%%%%%%%%%%%%%%%%%%%%%%%%%%%%%%%%%%%%%%%%%%%%%%%%%%
\begin{proof}
According to Ref.~\cite{roy2008exponentially}, we have that
\begin{subequations}
\begin{align}
    \ket{\psi_\theta^{1}}
    &= \cos(\theta^\prime)\ket{0}^{\otimes N}
    -i \sin(\theta^\prime)\ket{1}^{\otimes N},\\
    \ket{\psi_\theta^{2}}
    &= \cos(\theta^\prime)\ket{0}^{\otimes N}
    + \sin(\theta^\prime)\ket{1}^{\otimes N},
\end{align}
\end{subequations}
where $\theta^\prime = 2^{N-1}\theta$. To proceed, we need to evaluate all the terms $S_1(\theta), S_2(\theta), F_1(\theta)$, and $F_2(\theta)$. As a more general case, let us consider the $N$-qubit pure asymmetric GHZ state:
\begin{equation}
    \ket{\text{GHZ}_{\alpha,\,\beta}}
    = \alpha \ket{0}^{\otimes N}
    + \beta \ket{1}^{\otimes N},
\end{equation}
where $|\alpha|^2 + |\beta|^2 =1$ for complex coefficients $\alpha, \beta$. The reduced two-qubit state in any systems $i,j = 1,\ldots, N$ is given by
\begin{align}
    \nonumber
    \vr_{ij}^{(2)}
    &=\tr_{\overline{ij}}
    \left(
    \ket{\text{GHZ}_{\alpha,\,\beta}}
    \!\bra{\text{GHZ}_{\alpha,\,\beta}}
    \right)\nonumber
    \\
    &=\frac{1}{4}\left\{
    \eins_2^{\otimes 4}
    +\Delta [\sigma_z^{(i)}+\sigma_z^{(j)}]
    +\sigma_z^{(j)}\otimes \sigma_z^{(i)}
    \right\},
\end{align}
with $\Delta = |\alpha|^2 - |\beta|^2$. Thus we can immediately find
\begin{subequations}
\begin{align}
    S_1(\text{GHZ}_{\alpha,\,\beta})
    &= N\Delta^2,\\
    S_2(\text{GHZ}_{\alpha,\,\beta})
    &= \frac{N(N-1)}{2},\\
    F_1(\text{GHZ}_{\alpha,\,\beta})
    &= N\Delta^4,\\
    F_2(\text{GHZ}_{\alpha,\,\beta})
    &= \frac{3N(N-1)}{2}.
\end{align}
\end{subequations}
Substituting these into the form in Eq.~(\ref{eq:gainkcopy}) in the main text and taking the limit $\theta \to 0$, we can complete the proof.
\end{proof}
%%%%%%%%%%%%%%%%%%%%%%%%%%%%%%%%%%%%%%%%%%%%%%%%%%%%%%%%%%%%%%%%%%%%%%%%%%%%%%%%%%%%%%%%%%%%%%%%%

%%%%%%%%%%%%%%%%%%%%%%%%%%%%%%%%%%%%%%%%%%%%%%%%%%%%%%%%%%%%%%%%%%%%%%%%%%%%%%%%%%%%%%%%%%%%%%%%%
\twocolumngrid
%\bibliography{ref.bib}


\end{document}